\documentclass[envcountsame,runningheads]{llncs}
\usepackage{pscproc2}

\usepackage{cite}
\usepackage{etex}
\usepackage{amsmath}
\usepackage{amssymb}
\usepackage{graphicx}
\usepackage[labelsep=quad,indention=10pt]{subfig}
\usepackage{url}
\usepackage{todonotes}
\usepackage{multirow}
\usepackage{tikz}
\usepackage{booktabs}
\usepackage{tabularx}
\usepackage{colortbl}
\usetikzlibrary{matrix, arrows, automata, shapes, backgrounds, spy, calc, fit, fadings, decorations.text,shadows,shapes.geometric,decorations.pathreplacing,patterns}
\usepackage{listings}
\usepackage{wrapfig}

\newcommand*\rot{\rotatebox{90}}
\newcommand*\rotRight{\rotatebox{270}}

\usepackage[english,linesnumbered,ruled,noend]{algorithm2e}
\DontPrintSemicolon
\SetFuncSty{textsc}
\SetCommentSty{textit}
\SetDataSty{texttt}
\SetKwInOut{Initial}{Initial}
\SetKwInOut{Parameter}{Param.}
\SetKwInOut{Input}{Input}
\SetKwInOut{Data}{Data}
\SetKwInOut{Output}{Output}
\ResetInOut{Output} 
\let\oldnl\nl
\newcommand{\nonl}{\renewcommand{\nl}{\let\nl\oldnl}}

\newcommand{\UnaryOperator}[2][]{%
    \ifx&#1&%
    \ensuremath{\mathop{}\mathopen{}#2\mathopen{}}%
    \else%
    \ensuremath{\mathop{}\mathopen{}#2\mathopen{}\left(#1\right)}%
    \fi%
}
\newcommand{\UnaryArray}[2][]{%
  \ifx&#1&%
  \ensuremath{\mathop{}\mathopen{}#2\mathopen{}}%
  \else%
  \ensuremath{\mathop{}\mathopen{}#2\mathopen{}\lbrack#1\rbrack}%
  \fi%
}
\newcommand{\OpenInterval}[2][]{%
  \ifx&#1&%
  \ensuremath{\mathop{}\mathopen{}#2\mathopen{}}%
  \else%
  \ensuremath{\mathop{}\mathopen{}#2\mathopen{}[#1)}%
  \fi%
}

\newcommand{\SA}[1]{\UnaryArray[#1]{\mathsf{SA}}}
\newcommand{\LCP}[1]{\UnaryArray[#1]{\mathsf{LCP}}}
\newcommand{\lcp}[2]{\ensuremath{\mathop{\mathsf{lcp}}\left( #1,#2\right)}}
\newcommand{\RMQ}[2]{\UnaryArray[#1]{\mathsf{RMQ}_{#2}}}
\newcommand{\Text}[1]{\UnaryArray[#1]{\mathsf{T}}}
\newcommand{\Substring}[1]{\OpenInterval[#1]{\mathsf{T}}}
\newcommand{\ISA}[1]{\UnaryArray[#1]{\mathsf{ISA}}}

\newcommand{\PHI}[1]{\UnaryArray[#1]{\mathsf{PHI}}}
\newcommand{\DELTA}[1]{\UnaryArray[#1]{\mathsf{DELTA}}}
\newcommand{\buf}[1]{\UnaryArray[#1]{\mathsf{buf}}}

\newcommand{\PAb}[1]{\UnaryArray[#1]{\mathsf{PAb}}}
\newcommand{\ISAb}[1]{\UnaryArray[#1]{\mathsf{ISAb}}}

\newcommand{\BA}[1]{\UnaryArray[#1]{\ensuremath{\mathsf{BUCKET\_A}}}}
\newcommand{\BB}[1]{\UnaryArray[#1]{\ensuremath{\mathsf{BUCKET\_B}}}}
\newcommand{\BBS}[1]{\UnaryArray[#1]{\ensuremath{\mathsf{BUCKET\_BSTAR}}}}

\newcommand{\m}{\textsf{m}}
\newcommand{\n}{\textsf{n}}
\newcommand{\cZ}{\textsf{c0}}
\newcommand{\cO}{\textsf{c1}}
\newcommand{\TypeA}{\textmd{A}}
\newcommand{\TypeB}{\textmd{B}}
\newcommand{\TypeBs}{\ensuremath{\textmd{B}^\star}}

\newcommand{\geqlex}{\ensuremath{\geq}}
\newcommand{\leqlex}{\ensuremath{\leq}}
\newcommand{\llex}{\ensuremath{<}}
\newcommand{\glex}{\ensuremath{>}}
\newcommand{\eqlex}{\ensuremath{=}}

\newtheorem{observation}[theorem]{Observation}
\newcolumntype{C}[1]{>{\centering\let\newline\\\arraybackslash\hspace{0pt}}m{#1}}

\makeatletter
\def\textSq#1{%
\begingroup
\setlength{\fboxsep}{0.3ex}
\setbox1=\hbox{#1}
\setlength{\@tempdima}{\maxof{\wd1}{\ht1+\dp1}}
\setlength{\@tempdimb}{(\@tempdima-\ht1+\dp1)/2}
\raise-\@tempdimb\hbox{\fbox{\vbox to \@tempdima{%
  \vfil\hbox to \@tempdima{\hfil\copy1\hfil}\vfil}}}%
\endgroup%
}
\makeatother

\begin{document}
\title{Dismantling DivSufSort%
\thanks{This work was supported by the German Research Foundation (DFG), priority programme ``Algorithms for Big Data'' (SPP 1736).}}%
\author{Johannes Fischer \and Florian Kurpicz}
\authorrunning{J.~Fischer \and F.~Kurpicz}
\institute{
Dept.\ of Computer Science, Technische Universit{\"a}t Dortmund, Germany\\
\email{johannes.fischer@cs.tu-dortmund.de},
\email{florian.kurpicz@tu-dortmund.de}%
}
\maketitle

\begin{abstract}
We give the first concise description of
the fastest known suffix sorting algorithm in main memory, the \emph{DivSufSort}
by Yuta Mori. We then present an extension that also computes the LCP-array,
which is competitive with the fastest known LCP-array construction algorithm.
\end{abstract}


\section{Introduction} 
\label{sec:introduction}
The suffix array \cite{Manber1993} is arguably one of the most interesting and versatile data structure in stringology.
Despite the plethora of theoretical and practical papers on suffix sorting (see the two overview articles \cite{Puglisi2007,dhaliwal12trends} for an overview up to 2007/2012), 
the text indexing community faces the curiosity that the fastest and most space-conscious way to construct the suffix array is by an algorithm called \emph{DivSufSort} (coded by Yuta Mori), which has only appeared as (almost undocumented) source code, and has never been described in an academic context.
The speed and its space-consciousness make DivSufSort still the method of choice in many software systems, e.g.\ in bioinformatics libraries\footnote{\url{https://github.com/NVlabs/nvbio}, last seen 05.07.2017}, and in the \emph{succinct data structures library (sdsl)} \cite{Gog2014}.

The starting point of this article was that we wanted to get a better understanding of DivSufSort's functionality and the reasons for its advantages in performance, but we could not find any arguments for this neither in the literature nor in the documentation. We therefore dove into the source code (consisting of more than 1,000 LOCs) ourselves, and want to communicate our findings in this article.
We point out that just very recently
Labeit et al.~\cite{Labeit2016parallelWaveletTree} 
parallelized DivSufSort, making it also the fastest \emph{parallel} suffix array construction algorithm (on all instances but one).
We think that this successful parallelization adds another reason for why
a deeper study of DivSufSort is worthwile.

\subsubsection{Our Contributions and Outline.}
This article pursues two goals: First, it gives a concise description of the DivSufSort-algorithm (Sect.~\ref{sec:divsufsort}), so that readers wishing to understand or modify the source code have an easy-to-use reference at hand. Second (Sect.~\ref{sec:inducing_the_lcp}), we provide and describe our own enhancement of DivSufSort that also computes
related and equally important information, the array of longest common prefixes of lexicographically adjacent suffixes (\emph{LCP-array} for short). 
We test our implementation empirically on a well-accepted testbed and prove it competitive with existing implementations, sometimes even little faster.

To help the reader link our description to the implementation, we show relevant excerpts from the original code\footnote{\url{https://github.com/y-256/libdivsufsort}, last seen 05.07.2017}, along with their original line numbers in the source code (\textsf{difsufsort.c}, \textsf{sssort.c}, and \textsf{trsort.c}).
In the following, we use a \textsf{slanted} font for variables that also appear verbatim in the source code; e.g., $\textsf{T}$ for the text.

\section{Preliminaries} 
\label{sec:preliminaries}
Let $\Text{}=\Text{0}\Text{1}\dots\Text{\n-1}$ be a \emph{text} of length \n\ consisting of characters from an ordered \emph{alphabet} $\Sigma$ of size $\sigma=\vert\Sigma\vert$.
For integers $0\leq i\leq j\leq \n$, the notation $[i,j)$ represents the integers from $i$ to $j-1$, and $\Substring{i,j}$ the \emph{substring} $\Text{i}\dots\Text{j-1}$.
We call $S_i=\Substring{i,\n}$ the $i$-th \emph{suffix} of \Text{}.
The \emph{suffix array} \SA{} of a text \Text{} of length \n\ is a permutation of $[0,\n)$ such that $S_{\SA{i}}\llex S_{\SA{i+1}}$ for all $0\leq i< \n-1$.
In \SA{}, all suffixes starting with the same character $\cZ\in\Sigma$ form a contiguous interval called \cZ\emph{-bucket}.
The same is true for all suffixes starting with the same two characters $\cZ,\cO\in\Sigma$.
We call the corresponding intervals $(\cZ,\cO)$\emph{-buckets}.
The \emph{inverse suffix array} \ISA{} is the inverse permutation of \SA{}.
The \emph{longest common prefix} of two suffixes $S_i$ and $S_j$ is $\lcp{i}{j}=\max\left\{s\geq 0 \colon \Substring{i,i+s}\eqlex\Substring{j,j+s}\right\}$.
The \emph{longest common prefix array} \LCP{} of \Text{} contains the longest common prefixes of the lexicographically consecutive suffixes, i.e., $\LCP{0}=0$ and $\LCP{i}=\lcp{\SA{i-1}}{\SA{i}}$ for all $1\leq i\leq \n-1$.

\begin{figure}[t]
    \centering
    \begin{tabular}{lC{.575cm}C{.575cm}C{.575cm}C{.575cm}C{.575cm}C{.575cm}C{.575cm}C{.575cm}C{.575cm}C{.575cm}C{.575cm}C{.575cm}C{.575cm}}
        \toprule
        $i$&0&1&2&3&4&5&6&7&8&9&10&11&12\\
        \midrule
        $\mathsf{T}\lbrack i\rbrack$&\texttt{c}&\texttt{d}&\texttt{c}&\texttt{d}&\texttt{c}&\texttt{d}&\texttt{c}&\texttt{d}&\texttt{c}&\texttt{c}&\texttt{d}&\texttt{d}&\texttt{\$}\\
        \midrule
        $\mathop{\mathsf{type}}\left(i\right)$&\hphantom{$^{\star}$}\TypeBs&\TypeA&\hphantom{$^{\star}$}\TypeBs&\TypeA&\hphantom{$^{\star}$}\TypeBs&\TypeA&\hphantom{$^{\star}$}\TypeBs&\TypeA&\TypeB&\hphantom{$^{\star}$}\TypeBs&\TypeA&\TypeA&\TypeA\\
    \bottomrule
    \end{tabular}
    \caption{Classification of $\Text{}=\mathtt{cdcdcdcdccdd\$}$ (our running example).
    \label{fig:running}
    }
\end{figure}

We classify all suffixes as follows (a technique first introduced by \cite{Itoh1999}; see Figure~\ref{fig:running}).
The suffix $S_i$ is an \TypeA{}-suffix (or ``$S_i$ has type \TypeA{}'') if $\Text{i}\glex \Text{i+1}$ or $i = \n-1$.
If $\Text{i}\llex \Text{i+1}$, then $S_i$ is a \TypeB{}-suffix (or ``has type \TypeB{}'').
Last, if $\Text{i}\eqlex \Text{i+1}$ then $S_i$ has the same type as $S_{i+1}$.\footnote{This differs from \cite{Itoh1999}, where $S_{i}$ is always a \TypeB{}-suffix if $\Text{i}\eqlex \Text{i+1}$.}
We further distinguish \TypeB{}-suffixes: if $S_i$ has type \TypeB{} and $S_{i+1}$ has type \TypeA, then suffix $S_i$ is also a \TypeBs-suffix.
Note that there are at most $\frac{\n}{2}$ \TypeBs-suffixes.
The definition of types implies restrictions on how the suffixes are distributed within one bucket:
A $(\cZ,\cO)$-bucket cannot contain \TypeA-suffixes if $\cZ\llex \cO$, and it cannot contain \TypeB-suffixes if $\cZ\glex \cO$.
If $\cZ\eqlex \cO$ it cannot contain \TypeBs-suffixes.
The classification also induces a partial order among the suffixes (see also Fig.~\ref{fig:example_of_suffix_type_order}):
\begin{lemma}
    \label{lem:order_of_suffixes_by_type}
    Let $S_i$ and $S_j$ be two suffixes. Then
    \begin{enumerate}
        \item $S_i\llex S_j$ if $S_i$ has type \TypeA{}, $S_j$ has type \TypeB{} and $\Text{i}=\Text{j}$, and
        \item $S_i\llex S_j$ if $S_i$ has type \TypeBs{}, $S_j$ has type \TypeB{} but not type \TypeBs{} and $\Text{i,i+1}=\Text{j,j+1}$.
    \end{enumerate}
\end{lemma}
\begin{proof}
   \TypeA- and \TypeB-suffixes can only occur together in a $(\cZ,\cZ)$-bucket.
   Assume that $S_i$ and $S_j$ start with $\cZ\cZ$ followed by a (possibly empty) sequence of $\cZ$'s and $S_i, S_j$ have type \TypeA, \TypeB, resp.
   Let $u=\Text{i+\lcp{i}{j}}$ and $v=\Text{j+\lcp{i}{j}}$ be the first characters where the suffixes differ. Therefore, $u\leqlex \cZ$ and $v\geqlex \cZ$.
   Since the characters differ, at least one of the inequalities is strict.
   The argument for the second case works analogously.\qed
\end{proof}

\begin{figure}[t]
    \centering
    \begin{tikzpicture}
      \node at (0,0) {$(\mathtt{w,w})$};
      \node at (.9,0) {$(\mathtt{w,x})$};
      \node at (1.8,0) {$(\mathtt{w,y})$};
      \node at (2.7,0) {$(\mathtt{w,z})$};
      \node at (3.6,0) {$(\mathtt{x,w})$};
      \node at (4.5,0) {$(\mathtt{x,x})$};
      \node at (5.4,0) {$(\mathtt{x,y})$};
      \node at (6.3,0) {$(\mathtt{x,z})$};
      \node at (7.2,0) {$(\mathtt{y,w})$};
      \node at (8.1,0) {$(\mathtt{y,x})$};
      \node at (9,0) {$(\mathtt{y,y})$};
      \node at (9.9,0) {$(\mathtt{y,z})$};
      \node at (10.8,0) {$(\mathtt{z,w})$};
      \node at (11.7,0) {$(\mathtt{z,x})$};
      \node at (12.6,0) {$(\mathtt{z,y})$};
      \node at (13.5,0) {$(\mathtt{z,z})$};

      \draw (-.5,.25) -- (14,.25);
      \draw (3.15,.25) -- (3.15,-.25);
      \draw (6.75,.25) -- (6.75,-.25);
      \draw (10.35,.25) -- (10.35,-.25);
      \fill[preaction={fill, black!40},pattern=north east lines,draw] (-.45, -.25) rectangle (.45, -.7);
      \fill[preaction={fill, black!60},pattern=north west lines,draw] (.45, -.25) rectangle (.9, -.7);
      \fill[preaction={fill, black!40},pattern=north east lines,draw] (.9, -.25) rectangle (1.35, -.7);
      \fill[preaction={fill, black!60},pattern=north west lines,draw] (1.35, -.25) rectangle (1.8, -.7);
      \fill[preaction={fill, black!40},pattern=north east lines,draw] (1.8, -.25) rectangle (2.25, -.7);
      \fill[preaction={fill, black!60},pattern=north west lines,draw] (2.25, -.25) rectangle (3.15, -.7);

      \fill[preaction={fill, black!20},pattern=dots,draw] (3.15, -.25) rectangle (4.5, -.7);
      \fill[preaction={fill, black!40},pattern=north east lines,draw] (4.5, -.25) rectangle (4.95, -.7);
      \fill[preaction={fill, black!60},pattern=north west lines,draw] (4.95, -.25) rectangle (5.4, -.7);
      \fill[preaction={fill, black!40},pattern=north east lines,draw] (5.4, -.25) rectangle (5.85, -.7);
      \fill[preaction={fill, black!60},pattern=north west lines,draw] (5.85, -.25) rectangle (6.3, -.7);
      \fill[preaction={fill, black!40},pattern=north east lines,draw] (6.3, -.25) rectangle (6.75, -.7);

      \fill[preaction={fill, black!20},pattern=dots,draw] (6.75, -.25) rectangle (9, -.7);
      \fill[preaction={fill, black!40},pattern=north east lines,draw] (9, -.25) rectangle (9.45, -.7);
      \fill[preaction={fill, black!60},pattern=north west lines,draw] (9.45, -.25) rectangle (9.9, -.7);
      \fill[preaction={fill, black!40},pattern=north east lines,draw] (9.9, -.25) rectangle (10.35, -.7);

      \fill[preaction={fill, black!20},pattern=dots,draw] (10.35, -.25) rectangle (13.95, -.7);
    \end{tikzpicture}
    \caption{Position of the suffix types within the
    $\left(\cZ,\cO\right)$-buckets for
    $\Sigma=\left\{ \mathtt{w}, \mathtt{x}, \mathtt{y}, \mathtt{z} \right\}$.
    Light gray (\protect\tikz[baseline=.25ex]{ \fill[preaction={fill, black!20},pattern=dots,draw] (0, 0) rectangle (.3, .3); }) areas represent positions of \TypeA-suffixes,
    gray (\protect\tikz[baseline=.25ex]{ \fill[preaction={fill, black!40},pattern=north east lines,draw] (0, 0) rectangle (.3, .3); }) areas represent positions of \TypeB-suffixes, and
    dark gray (\protect\tikz[baseline=.25ex]{ \fill[preaction={fill, black!60},pattern=north west lines,draw] (0, 0) rectangle (.3, .3); }) areas represent positions of \TypeBs-suffixes.
    }
    \label{fig:example_of_suffix_type_order}
\end{figure}

Given two consecutive \TypeBs-suffixes $S_i$ and $S_j$ (i.e., there is no \TypeBs-suffix $S_k$ such that $i<k<j$), we call the substring $\Substring{i,j+2}$ \emph{\TypeBs-substring}.
Also, for the last \TypeBs-suffix $S_i$ (i.e., there is no \TypeBs-suffix $S_k$ with $i<k<\n$), the substring $\Substring{i,\n}$ is also called a \TypeBs-substring.

\section{DivSufSort} 
\label{sec:divsufsort}
In this section we describe DivSufSort based on its current implementation (libdivsufsort  v2.0.2).
The algorithm consists of three phases:
\begin{itemize}
  \item First, we identify the types of all suffixes and compute the corresponding $\cZ$- and $\left(\cZ,\cO\right)$-bucket borders.
  This requires one scan of the text.
  \item Next, we sort all \TypeBs-suffixes and place them at their correct position in \SA{}.
  This is the most complicated part, as we first have to sort the \TypeBs-substrings in-place.
  Then, we use the ranks of the sorted \TypeBs-substrings to sort the corresponding \TypeBs-suffixes.
  \item In the last step, we scan \SA{} twice to induce the correct position of all remaining suffixes.
  (We first scan from right to left to induce all \TypeB-suffixes, followed by a scan from left to right, inducing all \TypeA-suffixes.)
\end{itemize}

\begin{wrapfigure}{r}{0.28\textwidth}
  \centering
  \vspace{-1em}
  \includegraphics[scale=1.2]{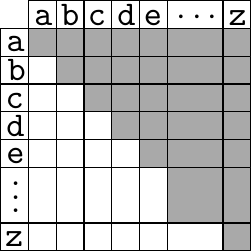}
  \caption{\BB{} (gray) and \textsf{BUCKET\_}-\textsf{BSTAR} represented as a 2-dimensional array.}
  \label{fig:bbs_2_dim_arr}
\end{wrapfigure}
Throughout the computation we utilize two additional arrays to store information about the buckets:
\BA{} (for \TypeA-suffixes) and \BB{} (for \TypeB- and \TypeBs-suffixes) of size $\sigma$ and $\sigma^2$, resp.
The former is used to store values associated with \TypeA-suffixes and is accessed by only one character.
The latter is used to store values associated with \TypeB- and \TypeBs-suffixes and is accessed by two characters.
\BB{\cZ,\cO} is short for \BB{\vert \cZ\vert\cdot\sigma+\vert \cO\vert} and \BBS{\cZ,\cO} is short for \BB{\vert \cO\vert\cdot\sigma+\vert \cZ\vert}, where $\vert \alpha\vert$ denotes the rank of $\alpha$ in the alphabet $\Sigma$. Information about both suffixes can be stored in the same array (Figure~\ref{fig:bbs_2_dim_arr}), as there are no \TypeBs-suffixes in $\left(\cZ,\cZ\right)$-buckets and no \TypeB-suffixes in $\left(\cZ,\cO\right)$-buckets for $\cZ\glex \cO$.
We denote the number of \TypeBs-suffixes by \m.

\subsection{Initializing DivSufSort} 
\label{sub:initializing_divsufsort}
\captionsetup*[subfigure]{position=bottom}
\begin{figure}[t]
    \centering
        \subfloat[]{%
            \begin{tabular}{rC{.355cm}C{.355cm}C{.355cm}C{.355cm}C{.355cm}C{.355cm}C{.355cm}C{.355cm}C{.355cm}C{.355cm}C{.355cm}C{.355cm}C{.355cm}}
                \arrayrulecolor{black!0} 
                \midrule
                \arrayrulecolor{black} 
                &&&&&&&&&&&&&\\
                \toprule
                    $i\hphantom{]}$&0&1&2&3&4&5&6&7&8&9&10&11&12\\
                \midrule
                $\mathsf{T}\lbrack i\rbrack$&\texttt{c}&\texttt{d}&\texttt{c}&\texttt{d}&\texttt{c}&\texttt{d}&\texttt{c}&\texttt{d}&\texttt{c}&\texttt{c}&\texttt{d}&\texttt{d}&\texttt{\$}\\
                \midrule
                    $\mathsf{SA}\lbrack i\rbrack$&0&0&0&0&0&0&0&0&\cellcolor{black!25}0&\cellcolor{black!25}2&\cellcolor{black!25}4&\cellcolor{black!25}6&\cellcolor{black!25}9\\
                \bottomrule
            \end{tabular}
            \label{fig:sa_first_scan}%
        }\hphantom{~}
        \subfloat[]{%
            \begin{tabular}{lccccc}
                \toprule
                                        &\texttt{\$}&\texttt{c}&\texttt{d} &(\texttt{c,c})&(\texttt{c,d})\\
                \midrule
                    \textsf{BUCKET\_A}    &1&0&6&-&-\\
                \midrule
                    \textsf{BUCKET\_B}    &-&-&-&1&-\\
                \midrule
                    \textsf{BUCKET\_BSTAR}   &-&-&-&-&5\\
                \bottomrule
            \end{tabular}
            \label{fig:buckets_first_scan}%
        }\\[.2em]
        \subfloat[]{%
            \begin{tabular}{rC{.355cm}C{.355cm}C{.355cm}C{.355cm}C{.355cm}C{.355cm}C{.355cm}C{.355cm}C{.355cm}C{.355cm}C{.355cm}C{.355cm}C{.355cm}}
                \arrayrulecolor{black!0} 
                \midrule
                \arrayrulecolor{black} 
                &\texttt{\$}&\multicolumn{6}{|c|}{\texttt{c}}&\multicolumn{6}{c}{\texttt{d}}\\
                \toprule
                    $i\hphantom{]}$&0&1&2&3&4&5&6&7&8&9&10&11&12\\
                \midrule
                $\mathsf{T}\lbrack i\rbrack$&\texttt{c}&\texttt{d}&\texttt{c}&\texttt{d}&\texttt{c}&\texttt{d}&\texttt{c}&\texttt{d}&\texttt{c}&\texttt{c}&\texttt{d}&\texttt{d}&\texttt{\$}\\
                \midrule
                    $\mathsf{SA}\lbrack i\rbrack$&0&0&0&0&0&0&0&0&\cellcolor{black!25}0&\cellcolor{black!25}2&\cellcolor{black!25}4&\cellcolor{black!25}6&\cellcolor{black!25}9\\
                \bottomrule
            \end{tabular}
          \label{fig:sa_prefix_sum}%
       }\hphantom{~}
       \subfloat[]{%
            \begin{tabular}{lccccc}
                \toprule
                                        &\texttt{\$}&\texttt{c}&\texttt{d} &(\texttt{c,c})&(\texttt{c,d})\\
                \midrule
                    \textsf{BUCKET\_A}    &\textbf{0}&\textbf{1}&\textbf{7}&-&-\\
                \midrule
                    \textsf{BUCKET\_B}    &-&-&-&1&-\\
                \midrule
                    \textsf{BUCKET\_BSTAR}   &-&-&-&-&5\\
                \bottomrule
            \end{tabular}
          \label{fig:buckets_prefix_sum}%
       }\\[.2em]
       \subfloat[]{%
            \begin{tabular}{rC{.355cm}C{.355cm}C{.355cm}C{.355cm}C{.355cm}C{.355cm}C{.355cm}C{.355cm}C{.355cm}C{.355cm}C{.355cm}C{.355cm}C{.355cm}}
                \arrayrulecolor{black!0} 
                \midrule
                \arrayrulecolor{black} 
                &\texttt{\$}&\multicolumn{6}{|c|}{\texttt{c}}&\multicolumn{6}{c}{\texttt{d}}\\
                \toprule
                    $i\hphantom{]}$&0&1&2&3&4&5&6&7&8&9&10&11&12\\
                \midrule
                $\mathsf{T}\lbrack i\rbrack$&\texttt{c}&\texttt{d}&\texttt{c}&\texttt{d}&\texttt{c}&\texttt{d}&\texttt{c}&\texttt{d}&\texttt{c}&\texttt{c}&\texttt{d}&\texttt{d}&\texttt{\$}\\
                \midrule
                    $\mathsf{SA}\lbrack i\rbrack$&\cellcolor{black!10}4&\cellcolor{black!10}0&\cellcolor{black!10}1&\cellcolor{black!10}2&\cellcolor{black!10}3&0&0&0&\cellcolor{black!25}0&\cellcolor{black!25}2&\cellcolor{black!25}4&\cellcolor{black!25}6&\cellcolor{black!25}9\\
                \bottomrule
            \end{tabular}
          \label{fig:sa_references}%
       }\hphantom{~}
       \subfloat[]{%
            \begin{tabular}{lccccc}
                \toprule
                                        &\texttt{\$}&\texttt{c}&\texttt{d} &(\texttt{c,c})&(\texttt{c,d})\\
                \midrule
                    \textsf{BUCKET\_A}    &0&1&7&-&-\\
                \midrule
                    \textsf{BUCKET\_B}    &-&-&-&1&-\\
                \midrule
                    \textsf{BUCKET\_BSTAR}   &-&-&-&-&\textbf{0}\\
                \bottomrule
            \end{tabular}
          \label{fig:buckets_references}%
       }
       \caption{%
       \SA{} and the buckets after the first scan of \Text{} are shown in \protect\subref{fig:sa_first_scan} and \protect\subref{fig:buckets_first_scan}.
       \PAb{} (dark gray $\color{black!25}{\blacksquare}$ in \protect\subref{fig:sa_first_scan}, \protect\subref{fig:sa_prefix_sum} and \protect\subref{fig:sa_references}) contains the text positions of all \TypeBs-suffixes in \emph{text order}.
       The buckets \protect\subref{fig:buckets_first_scan} contain the number of suffixes beginning with the corresponding characters.
       In \protect\subref{fig:buckets_prefix_sum}, they are updated such the first position of each $\cZ$-bucket is stored in \BA{\cZ} (bold entires).
       The \SA{} does not change during this update, see~\protect\subref{fig:sa_prefix_sum}.
       In \protect\subref{fig:sa_references} we stored references to the text positions in \SA{0..\m-1} (light gray $\color{black!10}{\blacksquare}$) and update the corresponding \BBS{} with the first position in \SA{0..\m-1} (bold entry in \protect\subref{fig:buckets_references}).
       }
       \label{fig:divsufsort_initialization}
\end{figure}

The initialization of DivSufSort is listed in \textsf{divsufsort.c}.
First, we scan \Text{} from right to left (line~60), determine the type of each suffix and store the sizes of the corresponding buckets in \BA{}, \BB{} and \BBS{} (lines~62, 69 and 65).
In addition, we store the text position of each \TypeBs-suffix at the end of \SA{} such that $\SA{}[\n-\m..\n)$ contains the text positions of all \TypeBs-suffixes (line~66).
We call this part of the suffix array \PAb{} with $\PAb{i}=\SA{\n-\m+i}$ for all $0\leq i<\m$ (line~94), see Figure~\ref{fig:divsufsort_initialization}~\subref{fig:sa_first_scan} and~\subref{fig:buckets_first_scan}.

Next (lines~81 to 90), we compute the prefix sum of \BA{} and\linebreak \BBS{}, such that \BA{\cZ} contains the leftmost position of each $\cZ$-bucket and \BBS{\cZ,\cO} contains the rightmost position of the corresponding \TypeBs-suffixes with respect only to other \TypeBs-suffixes, i.e., the positions are in the interval $\lbrack 0,\m)$, see Figures~\ref{fig:divsufsort_initialization}~\subref{fig:sa_prefix_sum}~and~\subref{fig:buckets_prefix_sum}, where \subref{fig:sa_prefix_sum} remains unchanged.
During the sorting step, we do not sort the text positions.
Instead we sort \emph{references} to these positions.
These references are stored in $\SA{}[0..\m)$ (line~97).
During this step, \BBS{\cZ,\cO} is updated (line~97), such that it now contains the leftmost reference corresponding to a \TypeBs-suffix in the $\left(\cZ,\cO\right)$-bucket within the interval $[0,\m)$.
The reference to the last \TypeBs-suffix is put at the beginning of its corresponding bucket (line~100).
This reference is a special case as it has no successor in \PAb{} that is required for the comparison of two \TypeBs-substrings, see Figure~\ref{fig:divsufsort_initialization}~\subref{fig:sa_references} and~\subref{fig:buckets_references}.


\subsection{Sorting the \TypeBs-Suffixes} 
\label{sub:sorting_the_b_star_suffixes}
In this section, we describe how the \TypeBs-suffixes are sorted in three steps.
First, all \TypeBs-substrings are sorted independently for each $(\cZ,\cO)$-bucket (lines~134 to 142) using functions defined in \textsf{sssort.c}.
Then (second step starting at line~146), a partial \ISA{} (named \ISAb{}) is computed, containing the ranks of the partially sorted \TypeBs-suffixes (sorted by their initial \TypeBs-substrings).
Using these ranks we compute the lexicographical order of all \TypeBs-suffixes adopting an approach similar to prefix doubling, in the last step using functions defined in \textsf{trsort.c} (line~159).
We augment the approach with \emph{repetition detection} as introduced by Maniscalco and Puglisi~\cite{Maniscalco2007}.

\subsubsection{Sorting the \TypeBs-Substrings.}
\label{ssub:sorting_the_b_star_substrings}
All \TypeBs-substrings in a \BBS{} are sorted independently and in-place.
The interval of \SA{} that has not been used yet ($\SA{}[\m..\n-\m)$) serves as a buffer during the sorting (line~133).
We refer to this part of \SA{} as \buf{} with $\buf{i}=\SA{\m+i}$ for all $0\leq i< \n-2\m$.
This part of DivSufSort can be executed in parallel by sorting the \BBS{} in parallel, i.e., all \TypeBs-substring in one \BBS{} are sorted sequentially, but multiple \BBS{} are processed in parallel (see \textsf{divsufsort.c}, lines~105 to 131).
Here, each process gets a buffer of size $\frac{\vert\buf{}\vert}{\textsf{p}}$, where \textsf{p} is the number of processes.
All following line numbers in this subsection refer to \textsf{sssort.c}.

In the default configuration we only sort 1024 elements at once (see \textsf{SS\_BLOCK}-\textsf{SIZE}, e.g., line~763).
If the size of \buf{} is smaller than 1024 or the size of the current bucket, the bucket is divided in smaller subbuckets which are then sorted and merged (see line~767, splitting due to the buffer size and the loop at line~770 splitting with respect to the number of elements).
Lines~789 to~802 are used to merge the last considered subbuckets.
If the currently sorted bucket contains the last \TypeBs-substring it is moved to the corresponding position (lines~811 and~813).

The heavy lifting is done by the function \textsf{ss\_mintrosort} that is an implementation of \emph{Introspective Sort} (ISS) \cite{Musser1997}.
It sorts all \TypeBs-substring within the interval $\lbrack\mathtt{first}, \mathtt{last}\rbrack$ (line~310).
ISS uses \emph{Multikey Quicksort} (MKQS) \cite{Bentley1997} and \emph{Heapsort} (HS).
MKQS is used $\lfloor\lg\left(\mathtt{last}-\mathtt{first}\right)\rfloor$ times to sort an interval before HS is used (if there are still elements in the interval that have been equal to the pivot each time, see line~333).
MKQS divides each interval into three subintervals with respect to a pivot element.
The first subinterval contains all substrings whose $k$-th character is smaller than the pivot, the second subinterval contains all substrings whose $k$-th character is equal to the pivot, and the last subinterval contains all substrings whose $k$-th character is greater than the pivot.
We call $k$ the \textsf{depth} of the current iteration (line~332).
ISS is not implemented recursively; instead, a stack is used to keep track of the unsorted subintervals and the smaller subintervals are always processed first.
This guarantees a maximum stack size of $\lg\ell$, where $\ell$ is the initial interval size \cite[p.~67]{Mehlhorn1984}.
The subintervals containing the substrings whose $k$-th character is not equal to the pivot are sorted using MKQS $\lfloor\lg\left(\mathtt{last}-\mathtt{first}\right)\rfloor$ times before using HS, where now  $\mathtt{last}$ and $\mathtt{first}$ refer to the first and last positions of these intervals (lines~414 and 428).

Whenever an unsorted (sub)bucket is smaller than a threshold ($8$ in the default configuration), \emph{Insertionsort} (IS) is used to sort the bucket and mark it sorted (line~326).
Whenever we compare two \TypeBs-Substrings during IS, we use the function \textsf{ss\_compare} that compares two \TypeBs-substrings starting at the current \textsf{depth} and compares the substrings character by character.

Throughout the sorting of the \TypeBs-substrings, substrings that cannot be fully sorted, i.e. \TypeBs-substrings that are equal, are marked by storing their bitwise negated reference (line~178).
Only the first reference of such an interval is stored normally to identify the beginning of an interval of unsorted substrings (line~178).
There are \TypeBs-suffixes that are not sorted completely by their initial \TypeBs-substrings e.g., in our example $\Text{}=\mathtt{cdcdcdcdccdd\$}$ the \TypeBs-substring \texttt{cdcd} occurs three times -- see Figure~\ref{fig:sorted_substrings_all}.
Therefore, we cannot determine the order of the corresponding \TypeBs-suffixes just using their initial \TypeBs-substring.
The idea of sorting the suffixes in a $(\cZ,\cO)$-bucket up to a certain \textsf{depth} is similar to the approach of Manzini and Ferragina \cite{Manzini2004}, who sort the suffixes up to a certain \LCP{}-value.


\subsubsection{Computing the Partial Inverse Suffix Array.} 
\label{ssub:computing_the_inverse_suffix_array}
After the \TypeBs-substrings are sorted, we compute the \ISA{} for the partially sorted \TypeBs-substrings (lines~146 to 156).
The inverse suffix array for the \TypeBs-suffixes is stored in $\SA{}[\m..2\m)$ and referred to as \ISAb{} with $\ISAb{i}=\SA{\m+i}$. \ISAb{i} contains the \emph{rank} of the $i$-th \TypeBs-suffix, i.e., the number of lexicographically smaller \TypeBs-suffixes.
All references to line numbers in this subsection refer to \textsf{divsufsort.c}.
We scan the $\SA{}[0..\m)$ from right to left (line~146) and distinguish between bitwise negated references (values $<0$, starting at line~154) and non-negated references (values $\geq0$, starting at line~147).
In the first case, we have reached an interval where we have references of suffixes which could not be sorted comparing only the \TypeBs-substring.
We assign each of those suffixes the greatest feasible rank, i.e., $\m-i$, where $i$ is the number of lexicographically greater suffixes (similar to Larsson and Sadakane~\cite{Larsson2007}).
In addition we also store the bitwise negation of the references, i.e., the original reference.
In the other case (a value $\geq 0$) we simply assign the correct rank to the \TypeBs-suffix.
Whenever we scan an interval of completely sorted \TypeBs-suffixes, we mark the first position of the interval in $\SA{}[0..\m)$ with $-k$, where $k$ is the size of the interval (line~150).
Now we can identify all sorted intervals as they start with a negative value whose absolute value is the length of the interval.

In our example (see Figure~\ref{fig:partial_isab}) we have two fully sorted intervals of length $1$ at \SA{0} and \SA{4}, and an only partially sorted interval in \SA{1..3}.

\begin{figure}[t]
    \centering
       \subfloat[]{%
            \begin{tabular}{rC{.355cm}C{.355cm}C{.355cm}C{.355cm}C{.355cm}C{.355cm}C{.355cm}C{.355cm}C{.355cm}C{.355cm}C{.355cm}C{.355cm}C{.355cm}}
                \arrayrulecolor{black!0} 
                \midrule
                &&&&&&&&&&&&&\\
                \midrule
                &&&&&&&&&&&&&\\
                \midrule
                \arrayrulecolor{black} 
                &\texttt{\$}&\multicolumn{6}{|c|}{\texttt{c}}&\multicolumn{6}{c}{\texttt{d}}\\
                \toprule
                    $i\hphantom{]}$&0&1&2&3&4&5&6&7&8&9&10&11&12\\
                \midrule
                $\mathsf{T}\lbrack i\rbrack$&\texttt{c}&\texttt{d}&\texttt{c}&\texttt{d}&\texttt{c}&\texttt{d}&\texttt{c}&\texttt{d}&\texttt{c}&\texttt{c}&\texttt{d}&\texttt{d}&\texttt{\$}\\
                \midrule
                    $\mathsf{SA}\lbrack i\rbrack$&\cellcolor{black!10}3&\cellcolor{black!10}0&\cellcolor{black!10}$\cellcolor{black!10}\tilde 1$&$\cellcolor{black!10}\tilde 2$&\cellcolor{black!10}4&0&0&0&\cellcolor{black!25}0&\cellcolor{black!25}2&\cellcolor{black!25}4&\cellcolor{black!25}6&\cellcolor{black!25}9\\
                \bottomrule
            \end{tabular}
          \label{fig:sa_sorted_substrings}%
       }\hphantom{~~~~~}
       \subfloat[]{%
            \begin{tabular}{lll}
                \toprule
                    Ref. & Text Pos. & \TypeBs-substring\\
                \midrule
                    3& 6   &\texttt{cdcc}\\
                \midrule
                    0& 0   &\texttt{cdcd}\\
                \midrule
                    1& 2   &\texttt{cdcd}\\
                \midrule
                    2& 4   &\texttt{cdcd}\\
                \midrule
                    4& 9   &\texttt{cdd\$}\\
                \bottomrule
            \end{tabular}
          \label{fig:sorted_substrings}%
       }
       \caption{The lexicographically sorted references of the \TypeBs-substrings in \SA{0..\m-1} (light gray $\color{black!10}{\blacksquare}$ in \protect\subref{fig:sa_sorted_substrings}). For readability we write $\tilde i$ if $i$ is bitwise negated ($\tilde i < 0$ for all $0\leq i\leq \n$). The content of the buckets
       is not changed in this step.
       The references, their corresponding text positions and the \TypeBs-substrings are shown in
       \protect\subref{fig:sorted_substrings}.}
       \label{fig:sorted_substrings_all}
\end{figure} 

\begin{figure}[t]
    \centering
            \begin{tabular}{rC{.5cm}C{.5cm}C{.5cm}C{.5cm}C{.5cm}C{.5cm}C{.5cm}C{.5cm}C{.5cm}C{.5cm}C{.5cm}C{.5cm}C{.5cm}}
                \toprule
                    $i\hphantom{]}$&0&1&2&3&4&5&6&7&8&9&10&11&12\\
                \midrule
                $\mathsf{T}\lbrack i\rbrack$&\texttt{c}&\texttt{d}&\texttt{c}&\texttt{d}&\texttt{c}&\texttt{d}&\texttt{c}&\texttt{d}&\texttt{c}&\texttt{c}&\texttt{d}&\texttt{d}&\texttt{\$}\\
                \midrule
                    $\mathsf{SA}\lbrack i\rbrack$&\cellcolor{black!10}\textbf{-1\hphantom{-}}&\cellcolor{black!10}0&\cellcolor{black!10}1&\cellcolor{black!10}2&\cellcolor{black!10}\textbf{-1\hphantom{-}}&\cellcolor{black!25}3&\cellcolor{black!25}3&\cellcolor{black!25}3&\cellcolor{black!25}0&\cellcolor{black!25}4&4&6&9\\
                \bottomrule
            \end{tabular}%
       \caption{\ISAb{} contains the inverse suffix array of the sorted \TypeBs-substrings. $\ISAb{i}=\SA{\m+i}$ for all $0\leq i< \m$ (dark gray $\color{black!25}{\blacksquare}$ in \SA{}). If $\m>\frac{\n}{3}$, \ISAb{} overlaps with \PAb{}. This does not matter, as we do not require the text positions at this point any more.
       While computing \ISAb{}, we also mark completely sorted intervals in \SA{0..\m-1}. The leftmost position of a sorted interval of length $\ell$ is changed to -$\ell$ (see \SA{0} and \SA{4} where we store -1 as the sorted intervals contain one entry).}
       \label{fig:partial_isab}
\end{figure}


\subsubsection{Sorting the \TypeBs-Suffixes.} 
\label{ssub:computing_the_isa_for_the_ssa}
In the last part of the \TypeBs-suffix sorting in DivSufSort we compute the correct ranks of all \TypeBs-suffixes and store them in \ISAb{}.
During this step, we only require information about the ranks of the suffixes and have no random access to the text, i.e., \PAb{} is not required any more.
All line numbers in this section refer to \textsf{trsort.c}.
Using \ISAb{}, we compute the ranks of all \TypeBs-suffixes using an approach similar to prefix doubling \cite{Larsson2007}.
Instead of doubling the length of the suffixes we double the number of considered \TypeBs-substrings that can have an arbitrary length (line~563).
Here, $\mathsf{ISAd}\lbrack i\rbrack$ refers to the rank of the $i+2^k$-th \TypeBs-suffix, where $k$ is the current iteration of the doubling algorithm.
Obviously, we need to update the ranks when we double the number of considered substrings, i.e., compute the new ranks for the \TypeBs-suffixes.
Since the ranks in the \ISA{} are given in text order, we can access the rank of the next (in text order) \TypeBs-substring for any given substring.

\subsubsection{Repetition Detection.}
The sorting that uses the new ranks as keys is done using Quicksort (QS), which also allows us to use the \textit{repetition} detection introduced by Maniscalco and Puglisi~\cite{Maniscalco2007} (see line~452 for the identification and the function \textsf{tr\_copy} for the computation of the correct ranks).
A repetition in \Text{} is a substring $\Text{i,i+rp}$ with $r\geq 2, p\geq 0$ and $i,i+rp\in\lbrack0,\n)$ such that $\Substring{i,i+p}=\Substring{i+p,i+2p}=\dots=\Substring{i+(r-1)p,i+rp}$.
Those repetitions are a problem if $S_{i}$ is a \TypeBs-suffix, since then $S_{kp}$ is a $B^{\star}$-suffix for all $k\leq r$.
We can simply sort all those suffixes by looking at the first character not belonging to the repetition ($\Text{i+rp+l}\neq\Text{i+l}$). If $\Text{i+rp+l}\llex\Text{i+l}$ then $\Text{i+(r-1)p+1,i+rp}\llex\Text{(i-1)+(r-1)p+1,(i-1)+rp}$ for all $1<i\leq r$.
The analogous case is true for $\Text{i+rp+l}\glex\Text{i+l}$, i.e., $\Text{i+(r-1)p+1,i+rp}\glex\Text{(i-1)+(r-1)p+1,(i-1)+rp}$ for all $1<i\leq r$.
This is done in lines~276~(and~282), where we increase (and decrease) the ranks of all suffixes in the repetition.
The identification of a repetition is supported by QS.
QS divides each interval into three subintervals (like MKQS).
We chose the median rank of the \TypeBs-suffixes that are considered during this doubling step as the pivot element for QS (line~455).
If the (current) rank of the first \TypeBs-suffix in the subinterval (considered in this doubling step) is equal to the pivot element, i.e., $\ISAb{i}=\mathsf{ISAd}\lbrack i\rbrack$ where $i$ is the first \TypeBs-suffix in the interval, then we have found a repetition (line~452, where \textsf{tr\_ilg} denotes the logarithm, i.e., the number of iterations until HS is used instead of QS).

Now we have computed the \ISA{} of all \TypeBs-suffixes (stored in \ISAb{}), i.e., we have all \TypeBs-suffixes in lexicographic order.
From this point on, all line numbers refer to \textsf{divsufsort.c}, again.
Next (see loop starting at line~162), we scan \Text{} from right to left, and when we read the $i$-th \TypeBs-suffix at position $j$, we store $j$ at position \SA{\ISAb{i}}.
Since we use the \TypeBs-suffixes to induce the \TypeB-suffixes (and we do not want to induce \TypeA-suffixes during the first inducing phase) we store the bitwise negation of $j$ if $S_{j-1}$ has type \TypeA{} (line~167).
Figures~\ref{fig:isab_fully_computed} and~\ref{fig:sa_final_text_positions} show the transition in $\SA{}[0..\m)$ for our example.
Now, $\SA{}[0..\m)$ contains the text positions of all \TypeBs-suffixes in lexicographic order.
Next (see loop beginning at line~173), we need to put these text positions at their correct position in $\SA{}[0..\n)$ (line~182).
While doing so, we update \BB{} and \BBS{} such that they contain the rightmost position of the corresponding buckets (lines~177 and~185).
Figures~\ref{fig:sa_moved_final_text_positions} and~\ref{fig:buckets_after_sorting_bs_suffixes} show this step for our running example.

\begin{figure}[t]
    \centering
       \subfloat[]{%
            \begin{tabular}{rC{.3275cm}C{.3275cm}C{.3275cm}C{.3275cm}C{.3275cm}C{.3275cm}C{.3275cm}C{.3275cm}C{.3275cm}C{.3275cm}C{.3275cm}C{.3275cm}C{.3275cm}}
                \toprule
                    $i\hphantom{]}$&0&1&2&3&4&5&6&7&8&9&10&11&12\\
                \midrule
                $\mathsf{T}\lbrack i\rbrack$&\texttt{c}&\texttt{d}&\texttt{c}&\texttt{d}&\texttt{c}&\texttt{d}&\texttt{c}&\texttt{d}&\texttt{c}&\texttt{c}&\texttt{d}&\texttt{d}&\texttt{\$}\\
                \midrule
                    $\mathsf{SA}\lbrack i\rbrack$&-1&-4&1&0&-1&\cellcolor{black!25}3&\cellcolor{black!25}2&\cellcolor{black!25}1&\cellcolor{black!25}0&\cellcolor{black!25}4&4&6&9\\
                \bottomrule
            \end{tabular}
          \label{fig:isab_fully_computed}%
       }
       \subfloat[]{%
            \begin{tabular}{C{.3275cm}C{.3275cm}C{.3275cm}C{.3275cm}C{.3275cm}C{.3275cm}C{.3275cm}C{.3275cm}C{.3275cm}C{.3275cm}C{.3275cm}C{.3275cm}C{.3275cm}}
                \toprule
                   0&1&2&3&4&5&6&7&8&9&10&11&12\\
                \midrule
                \texttt{c}&\texttt{d}&\texttt{c}&\texttt{d}&\texttt{c}&\texttt{d}&\texttt{c}&\texttt{d}&\texttt{c}&\texttt{c}&\texttt{d}&\texttt{d}&\texttt{\$}\\
                \midrule
                   \cellcolor{black!10}$\tilde 6$&\cellcolor{black!10}$\tilde 4$&\cellcolor{black!10}$\tilde 2$&\cellcolor{black!10}0&\cellcolor{black!10}9&\cellcolor{black!25}3&\cellcolor{black!25}2&\cellcolor{black!25}1&\cellcolor{black!25}0&\cellcolor{black!25}4&4&6&9\\
                \bottomrule
            \end{tabular}
          \label{fig:sa_final_text_positions}%
       }\\[.2em]
       \subfloat[]{%
            \begin{tabular}{rC{.355cm}C{.355cm}C{.355cm}C{.355cm}C{.355cm}C{.355cm}C{.355cm}C{.355cm}C{.355cm}C{.355cm}C{.355cm}C{.355cm}C{.355cm}}
                \arrayrulecolor{black!0} 
                \midrule
                \arrayrulecolor{black} 
                &\texttt{\$}&\multicolumn{6}{|c|}{\texttt{c}}&\multicolumn{6}{c}{\texttt{d}}\\
                \toprule
                    $i\hphantom{]}$&0&1&2&3&4&5&6&7&8&9&10&11&12\\
                \midrule
                $\mathsf{T}\lbrack i\rbrack$&\texttt{c}&\texttt{d}&\texttt{c}&\texttt{d}&\texttt{c}&\texttt{d}&\texttt{c}&\texttt{d}&\texttt{c}&\texttt{c}&\texttt{d}&\texttt{d}&\texttt{\$}\\
                \midrule
                    $\mathsf{SA}\lbrack i\rbrack$&$\tilde 6$&$\tilde 4$&$\cellcolor{black!10}\tilde 6$&\cellcolor{black!10}$\tilde 4$&\cellcolor{black!10}$\tilde 2$&\cellcolor{black!10}0&\cellcolor{black!10}9&1&0&4&4&6&9\\
                \bottomrule
            \end{tabular}
          \label{fig:sa_moved_final_text_positions}%
       }\hphantom{~}
       \subfloat[]{%
            \begin{tabular}{lccccc}
                \toprule
                                        &\texttt{\$}&\texttt{c}&\texttt{d} &(\texttt{c,c})&(\texttt{c,d})\\
                \midrule
                    \textsf{BUCKET\_A}    &0&1&7&-&-\\
                \midrule
                    \textsf{BUCKET\_B}    &-&-&-&\textbf{1}&\textbf{6}\\
                \midrule
                    \textsf{BUCKET\_BSTAR}   &-&-&-&-&\textbf{1}\\
                \bottomrule
            \end{tabular}
          \label{fig:buckets_after_sorting_bs_suffixes}%
       }
       \caption{\ISAb{} (dark gray $\color{black!25}{\blacksquare}$ in \protect\subref{fig:isab_fully_computed} and \protect\subref{fig:sa_final_text_positions}) contains the ranks of all \TypeBs-suffixes.
       The lexicographically sorted text positions of the \TypeBs-suffixes are shown light gray ($\color{black!10}{\blacksquare}$) in \protect\subref{fig:sa_final_text_positions}.
       Each text position $i$ is bitwise negated if $S_{i-1}$ has type \TypeA.
       In \protect\subref{fig:sa_moved_final_text_positions} all text positions of the \TypeBs-suffixes are at their correct position in \SA{0..\n-1} (light gray $\color{black!10}{\blacksquare}$).
       The buckets \protect\subref{fig:buckets_after_sorting_bs_suffixes} contain the leftmost position of the corresponding suffixes.}
\end{figure} 



\subsection{Inducing the $A$- and $B$-suffixes} 
\label{sub:inducing_the_a_and_b_suffixes}
Due to the types of the suffixes, we know that in any $(\cZ,\cO)$-bucket the \TypeA-suffixes are lexicographically smaller than the \TypeB-suffixes, and that \TypeBs-suffixes are lexicographically smaller than \TypeB-suffixes.
We also know that in lexicographic order, all consecutive intervals of \TypeB-suffixes are left of at least one \TypeBs-suffix and all \TypeA-suffixes are right of at least one \TypeB-suffix -- see Figure~\ref{fig:example_of_suffix_type_order}.
Now we scan \SA{} twice: once from right to left where all \TypeB-suffixes are induced (we can skip all parts of \SA{} containing only \TypeA-suffixes), and then from left to right to induce all \TypeA-suffixes (see Figure~\ref{fig:sa_inducing} for an example of the entire inducing process).
All following line numbers refer to \textsf{difsufsort.c}.
A step-by-step example is given in Figure~\ref{fig:sa_inducing}.

\begin{figure}
    \centering
    \begin{tabular}{rC{.45cm}C{.45cm}C{.45cm}C{.45cm}C{.45cm}C{.45cm}C{.45cm}C{.45cm}C{.45cm}C{.45cm}C{.45cm}C{.45cm}C{.45cm}C{.5cm}C{.4cm}C{.45cm}C{.45cm}C{.45cm}C{.45cm}C{.45cm}} 
        &&&&&&&&&&&&&&&\multirow{1}{*}{\rot{\BA{\texttt{\$}}\hphantom{\textsf{BSTR},\texttt{d}}}} & \multirow{1}{*}{\rot{\BA{\texttt{c}}\hphantom{\textsf{BSTR},\texttt{d}}}} & \multirow{1}{*}{\rot{\BA{\texttt{d}}\hphantom{\textsf{BSTR},\texttt{d}}}} & \multirow{1}{*}{\rot{\BB{\texttt{c},\texttt{c}}\hphantom{\textsf{STAR}}}} & \multirow{1}{*}{\rot{\BBS{\texttt{c},\texttt{d}}}} &\\
        &&&&&&&&&&&&&&&&&&\\
        &&&&&&&&&&&&&&&&&&\\
        &&&&&&&&&&&&&&&&&&\\
        &&&&&&&&&&&&&&&&&&\\
        &&&&&&&&&&&&&&&&&&\\
        &&\multicolumn{6}{|c|}{Scanned Interval}&&&&&&&&&&&&  \\
        \cmidrule{1-14}
        $i\hphantom{]}$&0&1&2&3&4&5&6&7&8&9&10&11&12&&&&&&&  \\
        \cmidrule[\heavyrulewidth]{1-14}\cmidrule[\heavyrulewidth]{16-21}
            $\mathsf{SA}\lbrack i\rbrack$&$\tilde 6$&$\tilde 4$&$\tilde 6$&$\tilde 4$&$\tilde 2$&0&9&1&0&4&4&6&9&&0&1&\textSq{7}&1&\textSq{1}&\multirow{8}{*}{\rotRight{\textbf{~First Induction Phase}}}\\
        \cmidrule{1-14}\cmidrule{16-20}
            $\mathsf{SA}\lbrack i\rbrack$&$\tilde 6$&\cellcolor{black!25}$\tilde 4$&$\tilde 6$&$\tilde 4$&$\tilde 2$&0&\cellcolor{black!10}9&1&0&4&4&6&9&&0&1&7&\cellcolor{black!25}1&1&\\
        \cmidrule{1-14}\cmidrule{16-20}
            $\mathsf{SA}\lbrack i\rbrack$&$\tilde 6$&$\mathbf{\tilde 8}$&$\tilde 6$&$\tilde 4$&$\tilde 2$&\cellcolor{black!10}0&$\mathbf{\tilde 9}$&1&0&4&4&6&9&&0&1&7&\textbf{0}&1&\\
        \cmidrule{1-14}\cmidrule{16-20}
            $\mathsf{SA}\lbrack i\rbrack$&$\tilde 6$&$\tilde 8$&$\tilde 6$&$\tilde 4$&$\cellcolor{black!10}\tilde 2$&$\mathbf{\tilde 0}$&$\tilde 9$&1&0&4&4&6&9&&0&1&7&0&1&\\
        \cmidrule{1-14}\cmidrule{16-20}
            $\mathsf{SA}\lbrack i\rbrack$&$\tilde 6$&$\tilde 8$&$\tilde 6$&$\cellcolor{black!10}\tilde 4$&\textbf{2}&$\tilde 0$&$\tilde 9$&1&0&4&4&6&9&&0&1&7&0&1&\\
        \cmidrule{1-14}\cmidrule{16-20}
            $\mathsf{SA}\lbrack i\rbrack$&$\tilde 6$&$\tilde 8$&$\cellcolor{black!10}\tilde 6$&\textbf{4}&2&$\tilde 0$&$\tilde 9$&1&0&4&4&6&9&&0&1&7&0&1&\\
        \cmidrule{1-14}\cmidrule{16-20}
            $\mathsf{SA}\lbrack i\rbrack$&$\tilde 6$&$\cellcolor{black!10}\tilde 8$&\textbf{6}&4&2&$\tilde 0$&$\tilde 9$&1&0&4&4&6&9&&0&1&7&0&1&\\
        \cmidrule{1-14}\cmidrule{16-20}
            $\mathsf{SA}\lbrack i\rbrack$&$\tilde 6$&\textbf{8}&6&4&2&$\tilde 0$&$\tilde 9$&1&0&4&4&6&9&&0&1&7&0&1&\\
        \cmidrule{1-14}\cmidrule{16-21}
        \addlinespace[.025em]
        \cmidrule{1-14}\cmidrule{16-21}
            $\mathsf{SA}\lbrack i\rbrack$&\cellcolor{black!25}\textbf{12}&8&6&4&2&$\tilde 0$&$\tilde 9$&1&0&4&4&6&9&&\cellcolor{black!25}0&1&7&0&1&\\
        \cmidrule{1-14}\cmidrule{16-21}
        \addlinespace[.025em]
        \cmidrule{1-14}\cmidrule{16-21}
            $\mathsf{SA}\lbrack i\rbrack$&\cellcolor{black!10}12&8&6&4&2&$\tilde 0$&$\tilde 9$&\cellcolor{black!25}1&0&4&4&6&9&&\textbf{1}&1&\cellcolor{black!25}7&0&1&\multirow{9}{*}{\rotRight{\textbf{~Second Induction Phase}}}\\
        \cmidrule{1-14}\cmidrule{16-20}
            $\mathsf{SA}\lbrack i\rbrack$&12&\cellcolor{black!10}8&6&4&2&$\tilde 0$&$\tilde 9$&\textbf{11}&\cellcolor{black!25}0&4&4&6&9&&1&1&\cellcolor{black!25}\textbf{8}&0&1&\\
        \cmidrule{1-14}\cmidrule{16-20}
            $\mathsf{SA}\lbrack i\rbrack$&12&8&\cellcolor{black!10}6&4&2&$\tilde 0$&$\tilde 9$&11&\textbf{7}&\cellcolor{black!25}4&4&6&9&&1&1&\cellcolor{black!25}\textbf{9}&0&1&\\
        \cmidrule{1-14}\cmidrule{16-20}
            $\mathsf{SA}\lbrack i\rbrack$&12&8&6&\cellcolor{black!10}4&2&$\tilde 0$&$\tilde 9$&11&7&\textbf{5}&\cellcolor{black!25}4&6&9&&1&1&\cellcolor{black!25}\textbf{10}&0&1&\\
        \cmidrule{1-14}\cmidrule{16-20}
            $\mathsf{SA}\lbrack i\rbrack$&12&8&6&4&\cellcolor{black!10}2&$\tilde 0$&$\tilde 9$&11&7&5&\textbf{3}&\cellcolor{black!25}6&9&&1&1&\cellcolor{black!25}\textbf{11}&0&1&\\
        \cmidrule{1-14}\cmidrule{16-20}
            $\mathsf{SA}\lbrack i\rbrack$&12&8&6&4&2&\cellcolor{black!10}$\tilde 0$&$\tilde 9$&11&7&5&3&\textbf{1}&9&&1&1&\textbf{12}&0&1&\\
        \cmidrule{1-14}\cmidrule{16-20}
            $\mathsf{SA}\lbrack i\rbrack$&12&8&6&4&2&\textbf{0}&\cellcolor{black!10}$\tilde 9$&11&7&5&3&1&9&&1&1&12&0&1&\\
        \cmidrule{1-14}\cmidrule{16-20}
            $\mathsf{SA}\lbrack i\rbrack$&12&8&6&4&2&0&\textbf{9}&\cellcolor{black!10}11&7&5&3&1&\cellcolor{black!25}9&&1&1&\cellcolor{black!25}12&0&1&\\
        \cmidrule{1-14}\cmidrule{16-20}
            $\mathsf{SA}\lbrack i\rbrack$&12&8&6&4&2&0&9&11&7&5&3&1&\textbf{10}&&1&1&\textbf{13}&0&1&\\
        \bottomrule
    \end{tabular}
    \caption{During the first phase, we induce \TypeB-suffixes and only scan intervals where \TypeB- and \TypeBs-suffixes occur.
    Each of those intervals ends left of the succeeding $\cZ$-bucket.
    Its borders are stored in the corresponding \BBS{} (boxed entries, the right border is not part of the interval).
    After the first phase we put the last suffix at the beginning of its corresponding bucket.
    During the second phase we scan the whole array, as we also store the bitwise negation of all entries that have already been used for inducing.
    The currently considered entry is marked light gray ($\color{black!10}{\blacksquare}$).
    The entries highlighted dark gray ($\color{black!25}{\blacksquare}$) are the positions where a value is induced. The bucket that contains the position is highlighted in the same color. Entries that have changed are bold in the following row.}
    \label{fig:sa_inducing}%
\end{figure}

During the inducing of the \TypeB-suffixes, i.e., the first scan of \SA{} (see loop starting at line~205), whenever we read an entry $i$ in \SA{} such that $i > 0$ (line~211), we store the entry $i-1$ at the rightmost free position (a position in which a correct text position has not been stored yet) in the $(\Text{i-1},\Text{i})$-bucket (line~220).
If $\Text{i-2}\glex\Text{i-1}$, then $S_{i-2}$ is an \TypeA-suffix, which is not induced during the first scan, but the bitwise negated value of $i-1$ is stored instead (line~217).
Every position is overwritten with its bitwise negated value.
If the position was already bitwise negated, i.e., it has been induced and the corresponding suffix has type \TypeA{}, it is considered during the next scan (line~226) and it is ignored otherwise.
After the first traversal, all suffixes that have been used for inducing are represented by their bitwise negated position whereas all other suffixes are represented by their position, i.e., a positive integer.
It should be noted that all induced suffixes are lexicographically smaller than the suffix they are induced from: if we induce from a $\left(\cZ,\cO\right)$-bucket, we know that $\cZ\leqlex \cO$, since we are considering \TypeB-suffixes.
In addition, we can only induce in $\left(\cZ,\cO\right)$-buckets with $\cO\leqlex \cZ$, as only \TypeB-suffixes are considered during this traversal.

Before \SA{} is scanned a second time, $\n-1$ is stored at the beginning of the \Text{\n-1}-bucket (line~234).
If $S_{\n-2}$ has type \TypeA, we store $\n-1$ (we want to induce $S_{\n-2}$ during the second scan).
Otherwise, we store the bitwise negation of $\n-1$.

\begin{wrapfigure}{r}{0.585\textwidth}
  \centering
  \includegraphics[scale=1.2]{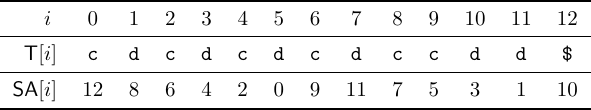}
  \caption{The final \SA{} of $\Text{}=\mathtt{cdcdcdcdccdd\$}$.}
  \label{fig:final_sa}
\end{wrapfigure}
During the second scan of \SA{} (see loop starting at line~236), whenever an entry $i$ of \SA{} is smaller than 0 it is overwritten by its bitwise negated value, i.e., the position of the suffix in the correct position in the suffix array (line~249).
Whenever $i>0$ (line~237) the suffix $S_{i-1}$ is induced at the leftmost free position in the $\Text{i-1}$-bucket (line~243).
Since all remaining suffixes are induced during this scan it is sufficient to identify the border using the $\cZ$-buckets, i.e., the value stored in \BA{\cZ}.
If the induced suffix would induce a \TypeB-suffix, its bitwise negated value is induced instead (line~240).
At the end of the traversal \SA{} contains the indices of all suffixes in lexicographic order.


\section{Inducing the \LCP{}-Array} 
\label{sec:inducing_the_lcp}
We now show how to modify DivSufSort such that it also computes the \LCP{}-array in addition to \SA{}.
To do so, we extend DivSufSort at three points of the computation of \SA{}.
First, we need to compute the \LCP{}-values of all \TypeBs-suffixes.
Next, during the inducing step, we also induce the \LCP{}-values for \TypeA- and \TypeB-suffixes.
For this we utilize a technique also described in~\cite{Fischer2011,Bingmann2013a} that allows us to answer \RMQ{}{}s on \LCP{} using only a stack~\cite{Gog2011}.
Last, we compute the \LCP{}-values of suffixes at the border of buckets, as those values cannot be induced.

Recall that the \LCP{}-value of two arbitrary suffixes $S_i$ and $S_j$ is denoted by \lcp{i}{j}. We need the following additional definition:
Given an array $A$ of length $\ell$ and $0\leq i\leq j \leq\ell$, a \emph{range minimum query} \RMQ{i,j}{A} asks for the minimum in $A$ in the interval $\lbrack i,j\rbrack$, in symbols: $\RMQ{i,j}{A}=\min\left\{A[k]\colon i\leq k\leq j\right\}$.

\subsection{Computing the \LCP{}-Values of the \TypeBs-Suffixes} 
\label{sub:sorting_the_b_star_suffixes_and_compute_their_lcp_values}
During the sorting of the \TypeBs-suffixes (right before the \TypeBs-suffixes are put at their correct position in $\SA{}[0..\n)$), all lexicographically sorted \TypeBs-suffixes are in $\SA{}[0..\m)$.
There are two cases regarding $\m$ (the number of \TypeBs-suffixes).
If $\m>\frac{\n}{3}$, we have overwritten the text positions of the \TypeBs-suffixes in \PAb{} with \ISAb{}.
In this case we must compute the \LCP{}-values naively.\footnote{For all tested instances (see Section~\ref{sec:experiments}) $\m\leq\frac{\n}{3}$.}
Otherwise (we still know the text positions of all \TypeBs-suffixes), we compute their \LCP{}-values using a sparse version of the $\Phi$-algorithm~\cite{Karkkainen2009a}, based on Observation~\ref{obs:example_for_obs_lcp}, which was also used implicitly in~\cite{Fischer2011,Bingmann2013a}.

\begin{observation}
   \label{obs:example_for_obs_lcp}
    If $S_i,S_{i^{\prime}},S_j$ and $S_{j^{\prime}}$ are \TypeBs-suffixes such that $i<i^{\prime},j<j^{\prime}$ and there is no other $B^{\star}$-suffix $S_{k}$ such that $i<k<i^{\prime}$ or $j<k<j^{\prime}$, then $\lcp{i^{\prime}}{j^{\prime}}\geq\lcp{i}{j}-(i^{\prime} - i)$.
\end{observation}

This is possible as we know the distance (in the text) of two \TypeBs-suffixes, i.e., $\PAb{i}-\PAb{j}$ is the distance of the $i$-th and $j$-th \TypeBs-suffix with $1\leq i\leq j\leq \m$.
See Figure~\ref{fig:example-phi} for an Example.
Algorithm~\ref{alg:sparse-lcp} shows the \emph{sparse} version of the $\Phi$-algorithm.
The difference to the original algorithm \cite{Karkkainen2009a} is that the next considered suffix is an arbitrary number of character shorter than the previous one, which results in Observation~\ref{obs:example_for_obs_lcp}.
The computation of the \LCP{}-values does not require any additional memory except for the $\n$ words for \LCP{}, where we temporarily store additional data.

First (lines~\ref{alg:sparse-lcp-fill-phi-begin} to \ref{alg:sparse-lcp-fill-phi-end} of Algorithm~\ref{alg:sparse-lcp}), we fill the \PHI{} (stored in $\LCP{}[\m..2\m)$) such that $\PHI{i}$ contains the text position of the suffix that is lexicographically consecutive to the $i$-th suffix (text position).
In $\DELTA{i}$ (stored in $\LCP{}[\n-\m..\n)$) we store the text distance of the $i$-th and $(i+1)$-th \TypeBs-suffix (text occurrence), i.e., $\PAb{i+1}-\PAb{i}$.
Then (lines \ref{alg:sparse-lcp-computation-begin} to \ref{alg:sparse-lcp-computation-end}), we compute the sparse \LCP{}-array using Observation~\ref{obs:example_for_obs_lcp}.
As we store the \LCP{}-values in \PHI{} in text order, we need to rewrite them to \LCP{} (line~\ref{alg:sparse-lcp-computation-rewrite}).

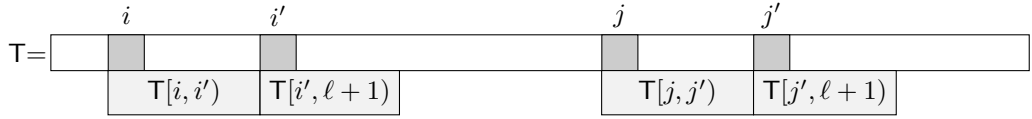
\begin{figure}[t]
   \centering
   \begin{tikzpicture}
      \node[rectangle split, rectangle split parts=30, rectangle split horizontal, draw, anchor=center, rectangle split draw splits=false, minimum height=.475cm] (t1) at (0,0) { };
      \node[left of=t1, node distance=6.75cm] {$\Text{}=$};
      \node[rectangle, draw, fill=black!20, minimum size=.475cm] (i) at ($(t1.one west) + (1, 0)$) {};
      \node[rectangle, draw, fill=black!20, minimum size=.475cm] (ip) at ($(t1.one west) + (3, 0)$) {};
      \node[rectangle, draw, fill=black!20, minimum size=.475cm] (j) at ($(t1.one west) + (7.5, 0)$) {};
      \node[rectangle, draw, fill=black!20, minimum size=.475cm] (jp) at ($(t1.one west) + (9.5, 0)$) {};

      \node[above of=i, node distance=1.4em] {$i$};
      \node[above of=ip, node distance=1.4em] {$i^{\prime}$};
      \node[above of=j, node distance=1.4em] {$j$};
      \node[above of=jp, node distance=1.4em] {$j^{\prime}$};

      \node[rectangle, draw, fill=black!5, minimum size=.475cm, minimum width=2cm] (lcp-i-j-a) at ($(t1.one west) + (1.76, -.53)$) {\small $\Substring{i,i^{\prime}}$};
      \node[rectangle, draw, fill=black!5, minimum size=.475cm, minimum width=2cm] (lcp-i-j-b) at ($(t1.one west) + (8.26, -.53)$) {\small $\Substring{j,j^{\prime}}$};
      \node[rectangle, draw, fill=black!5, minimum size=.475cm, minimum width=1.5cm] (lcp-ip-jp-a) at ($(t1.one west) + (3.68, -.53)$) {\small $\Substring{i^{\prime},\ell+1}$};
      \node[rectangle, draw, fill=black!5, minimum size=.475cm, minimum width=1.5cm] (lcp-ip-jp-b) at ($(t1.one west) + (10.2, -.53)$) {\small $\Substring{j^{\prime},\ell+1}$};

    \end{tikzpicture}
    \caption{Let $S_i,S_j,S_{i^\prime}$ and $S_{j^\prime}$ be \TypeBs-suffixes
    such that there is no \TypeBs-suffix $S_k$ with $i<k<i^\prime$ or 
    $j<k<j^\prime$, and let the \LCP{}-value of $S_i$ and $S_j$ be $\ell=\lcp{i}{j} + i$.
    Then the \LCP{}-value of
    $S_{i^{\prime}}$ and $S_{j^{\prime}}$ is $\lcp{i^\prime}{j^\prime}=\ell-i^{\prime}=\lcp{i}{j} - \left(i^\prime-i\right)$.
    }
   \label{fig:example-phi} 
\end{figure}

\begin{algorithm}[b]
   \caption[]{Sparse $\Phi$-Algorithm}
   \label{alg:sparse-lcp}
    \Input{$\Text{}$, $\m$, $\SA{}$, $\ISAb{}=\SA{\m..2\m-1}$, $\PAb{}=\SA{\n-\m..\n-1}$ and
    $\LCP{}$, $\PHI{}=\LCP{\m..2\m-1}$ $\PHI{}=\DELTA{\n-\m..\n-1}$.}
   \Output{\LCP{0..\m-1} contains the \LCP{}-values of the \TypeBs-suffixes.}
    $\PHI{\SA{0}} = -1$\;\label{alg:sparse-lcp-fill-phi-begin}
   \For{$i= 1;~i\leq \m-1;~i= i+1$} 
   {
      $\PHI{\SA{i}}= \SA{i - 1}$\;
      $\DELTA{i-1}= \PAb{i}-\PAb{i+1}$\label{alg:sparse-lcp-fill-phi-end}
   }
   \For{$i = 0,~p = 0;~i < m;~i = i+1$\label{alg:sparse-lcp-computation-begin}}
   {
       \While{$\Text{\PAb{i}+p+1} = \Text{\PAb{\PHI{i}}+p+1}$}{$p= p+1$\;}
        $\PHI{i} = p$ and $p= \max \left\{ 0, p-\DELTA{i} \right\}$\label{alg:sparse-lcp-computation-end}
   }
   \lFor{$i = 0;~i < m;~i= j+1$} 
   {
       $\LCP{\ISAb{i}}= \PHI{i};$\label{alg:sparse-lcp-computation-rewrite}
   }
\end{algorithm}

\subsection{Inducing the \textsc{LCP}-Values in Addition to the \textsc{SA}} 
\label{sub:inducing_the_lcp_values_in_addition_to_the_sa}
During the inducing of the \TypeB-suffixes, whenever a suffix is induced at position $u$ in \SA{} and there is already a suffix at position $u+1$ in the same $\left(\cZ,\cO\right)$-bucket, there are two cases:
\begin{enumerate}
  \item The suffixes $S_{\SA{u}}$ and $S_{\SA{u+1}}$ have been induced from suffixes $S_{\SA{v}},S_{\SA{w}}$ in the same $\left(\cZ,\cO\right)$-bucket; in this case $\LCP{u+1}=\RMQ{v+1,w}{\LCP{}}+1$.
  \item Otherwise, the \LCP{}-value is either $1$ or $2$, depending on the $\cZ$-buckets $S_{\SA{v}},S_{\SA{w}}$ are. If they are in the same bucket the \LCP{}-value is $2$ and $1$ if not.
 \end{enumerate}
The computation of the \LCP{}-values during the inducing of the \TypeA-suffixes works analogously.
This leads to the following observation for the general case:
\begin{observation}
    Let $\SA{u}=i,\SA{u+1}=j,\SA{v}=i+1$ and $\SA{w}=j+1$ such that
    $S_i$ and $S_j$ are in the same \cZ-bucket and $u+1<v,w$ or $w,v<u$.
    Then $\LCP{u+1}=\RMQ{\min\left\{ v,w \right\} + 1,\max \left\{ v,w \right\}}{\LCP{}}+1$.
\end{observation}

Not all \LCP{}-values can be induced this way. The missing cases are covered in the next section.
Instead of using a dynamic \RMQ{}{} data structure, we can answer the \RMQ{}{}s using a  \emph{min-stack} \cite{Bingmann2013a,Fischer2011,Gog2011}.
We only need to consider \RMQ{}{}s for suffixes from the same $\left(\cZ,\cO\right)$-bucket.
To this end, we build the min-stack while scanning an interval $[\mathsf{first},\mathsf{last}]$ (from right to left) of the \LCP{}-array.
An entry on the min-stack consist of tuple $\langle k,\LCP{k}\rangle$.
Initially, the tuple $\langle \n,-1\rangle$ is on the min-stack.
To update the min-stack at position $i\in[\mathsf{first},\mathsf{last}]$ we look at the top of the min-stack and remove the tuple $\langle k,\LCP{k}\rangle$ if $\LCP{k}\geq\LCP{i}$.
We repeat this process until no tuple is removed.
Then we add $\langle i,\LCP{i}\rangle$ to the min-stack.

Now we want to answer $\RMQ{i,j}{\LCP{}}$ with $\mathsf{first}\leq i<j\leq\mathsf{last}$.
(It should be noted that at this point we have not added $\langle i,\LCP{i}\rangle$ to the min-stack or have removed any tuple from the min-stack in the process of adding it to the min-stack.)
To this end, we scan the min-stack from top to bottom, until we find two consecutive tuples $\langle k,\LCP{k}\rangle$, $\langle k^{\prime},\LCP{k^{\prime}}\rangle$ such that $k^{\prime}>j$.
Then, $\RMQ{i,j}{\LCP{}}=\LCP{k}$.
If we scan from left to right, the min-stack works analogously.
The only difference is that the initial tuple is $\langle -1,-1\rangle$ and we search for the two consecutive tuples until $k^{\prime}<j$

The min-stack is reseted whenever we arrive at a new $\left(\cZ,\cO\right)$-bucket, i.e., we only keep the $\langle \n,-1\rangle$-tuple.
In the implementation, the min-stack is realized using a single array and a reference to its current top.

\begin{figure}
    \centering
    \subfloat[]{%
      \begin{tabular}{rccccccc}
      &&&&&&&\\
      \toprule
      $i\hphantom{]}$ & 0 & 1 & 2 & 3 & 4 & 5 & 6\\
      \midrule
      $\mathsf{A}\lbrack i\rbrack$ & 4 & 2 & 0 & 1 & 4 & 3 & 2 \\
      \bottomrule
      &&&&&&&\\
      \end{tabular}
      \label{fig:min_stack_array}%
    }~~~
    \subfloat[]{%
      \begin{tabular}{cccccccc}
      \arrayrulecolor{black!0} 
      \toprule
      \hphantom{$\langle 4,4\rangle$}& &  & $\langle 4,4\rangle$ &  &  &  & \\
      & & $\langle 5,3\rangle$ & $\langle 5,3\rangle$ &  &  & $\langle 1,2\rangle$ & $\langle 1,2\rangle$\\
      & $\langle 6,2\rangle$ & $\langle 6,2\rangle$ & $\langle 6,2\rangle$ & $\langle 3,1\rangle$ & $\langle 2,0\rangle$ & $\langle 2,0\rangle$ & $\langle 2,0\rangle$\\
      & $\langle \n,-1\rangle$ & $\langle \n,-1\rangle$ & $\langle \n,-1\rangle$ & $\langle \n,-1\rangle$ & $\langle \n,-1\rangle$ & $\langle \n,-1\rangle$ & $\langle \n,-1\rangle$\\
      \arrayrulecolor{black} 
      \bottomrule
      $i$ & 6 & 5 & 4 & 3 & 2 & 1 & 0\\
      \end{tabular}
      \label{fig:min_stack_stack}%
    }
    \caption{The min-stack for each \emph{current} position $i$
    \protect\subref{fig:min_stack_stack} while scanning $\mathsf{A}$
    \protect\subref{fig:min_stack_array} from right to left. A tuple
    $\left(p, v\right)$ contains the position $p$ of the value $v$.
    For the current position $i$ the stack can be used to answer \RMQ{}{}s
    of the type \RMQ{i,j}{\mathsf{A}} with $j\geq i$ by looking at elements
    from the top until a position $k$ with $k\geq j$ is found.}
    \label{fig:min_stack}
\end{figure}

In addition to the min-stack, we require for each $\cZ$-bucket the position of where the last suffix has been induced from.
This is the position we look for when querying the min-stack.


\subsection{Special Cases during \LCP{} Induction} 
\label{sub:special_cases_during_lcp}
There are three special cases where the \LCP{}-value cannot be induced using the min-stack (or \RMQ{}{}s in general).
The first case occurs if a suffix is induced next to a \TypeBs-suffix.
The  inducing can happen to the left or right of the already placed \TypeBs-suffix.
The former case is easy as there cannot be an \TypeA- or \TypeB-suffix to the left of a \TypeBs-suffix in the same $\left(\cZ,\cO\right)$-bucket.
Therefore, we only need to check whether the suffixes are in the same $\cZ$-bucket to compute the \LCP{}-value for the \TypeBs-suffix, which is either $0$ or $1$.
The other case (a suffix is induced to the right of a \TypeBs-suffix) is more demanding, as the \LCP{}-value must be computed.
Fortunately, this can be done more sophisticated than by naive comparison of the suffixes.
First, we check whether both the \TypeBs-suffix $S_i$ and the \TypeB-suffix $S_j$ are in the same $\left(\cZ,\cO\right)$-bucket.
If not, the \LCP{}-value is $1$ if they occur in the same $\cZ$-bucket, and $0$ otherwise.
However, if they occur in the same $\left(\cZ,\cO\right)$-bucket, we know that $S_i$ has a prefix $\cZ\cO d$, $d\in\Sigma$, such that $\cZ\llex \cO\geqlex d$, and that $S_j$ has a prefix $\cZ\cO e$, $e\in\Sigma$, such that $\cZ\llex \cO\leqlex e$.
Hence, the \LCP{}-value is $\max \left\{ k\geq 0 \colon \Substring{i+1,i+k+2} = \Substring{j+1,j+k+2} \right\} + 1$, i.e., the first appearance of a character not equal to $\cO$ in either suffix.
In the last case (an \TypeA-suffix is induced next to a \TypeB-suffix) the \LCP{}-value can be determined in an analogous way.


\section{Experiments with LCP-Construction}
\label{sec:experiments}

We implemented the modified DivSufSort in \textsc{C} and compiled it using \texttt{gcc} version 6.2 with the compiler options \textsf{-DNDEBUG}, \textsf{-03} and \textsf{-march=native}.
Our implementation is available from \url{https://github.com/kurpicz/libdivsufsort}.
We ran all experiments on a computer equipped with an Intel Core i5-4670 processor and 16~GiB RAM, using only a single core.

We evaluated our algorithm on the \emph{Pizza \& Chili} Corpus\footnote{\url{http://pizzachili.dcc.uchile.cl/}, last seen 05.07.2017} and compared our implementation to the following \LCP{}-construction algorithms (using the same compiler options):
\emph{KLAAP}~\cite{Kasai2001} is the first linear-time \LCP{}-construction algorithm.
The $\Phi$-algorithm~\cite{Karkkainen2009a} is an alternative to KLAAP that reduces cache-misses.
\emph{Inducing+SAIS}~\cite{Fischer2011} is an \LCP{}-construction algorithm (using similar ideas as in this paper) based on SAIS~\cite{Nong2009}, and \emph{naive} scans the suffix array and checks two consecutive suffixes character by character.

We also looked at \LCP{}-construction algorithms requiring the \emph{Burrows-Wheeler transform}, i.e., \emph{GO} and \emph{GO2} by Gog and Ohlebusch~\cite{Gog2011}.
Since these algorithms are only available in the \emph{succinct data structure library} (SDSL)~\cite{Gog2014}, which has an emphasis on a low memory footprint, the running times are affected by that.

The results of our experiments can be found in Table~\ref{tab:time_evaluation}.
As a brief summary, our practical tests show that $\Phi$ (see column 1) is the fastest \LCP{}-construction algorithm if \SA{} is already given, while our new implementation (column 6) is faster than the only other inducing-based approach (last 2 columns).

\newcommand{\specialcell}[2][c]{%
\begin{tabular}[#1]{@{}c@{}}#2\end{tabular}}

\begin{table}[t]
\setlength{\tabcolsep}{4pt}
\centering
\begin{tabular}{ccccccccccccc}
\toprule
\multicolumn{2}{c}{} & \multicolumn{7}{c}{\textsf{LCP} given \textsf{SA} (and \textsf{BWT} if necessary)} & \multicolumn{2}{c}{\textsf{SA}} & \multicolumn{2}{c}{$\mathsf{SA}+\mathsf{LCP}$} \\
\cmidrule(r){3-9}
\cmidrule(r){10-11}
\cmidrule(r){12-13}
\multicolumn{2}{c}{Text} & \rot{$\Phi$~\cite{Karkkainen2009a}} & \rot{KLAAP~\cite{Kasai2001}} & \rot{naive} & \rot{GO~\cite{Gog2011}}& \rot{GO2~\cite{Gog2011}} & \rot{\specialcell{inducing\\ $\lbrack$this paper$\rbrack$}} & \rot{inducing~\cite{Fischer2011}} & \rot{DivSufSort} & \rot{SAIS~\cite{Nong2009}} & \rot{\specialcell{inducing +\\DivSufSort\\$\lbrack$this paper$\rbrack$}} & \rot{\specialcell{inducing +\\SAIS~\cite{Fischer2011}}} \\
\cmidrule[0.225ex](r){1-2}
\cmidrule[0.225ex](r){3-9}
\cmidrule[0.225ex](r){10-11}
\cmidrule[0.225ex](r){12-13}
\multirow{5}{*}{\rotatebox[origin=c]{90}{20~MB}} & dna       &\textbf{0.77}& 0.91& 1.180&6.46&2.65& 0.78& 1.12& 1.45& 1.71& 2.23& 2.83\\
                                                 &  english  &\textbf{0.61}& 0.77& 44.72&7.90&4.03& 0.64& 0.91& 1.45& 1.65& 2.09& 2.56\\
                                                 &  dblp.xml &0.54& 0.55& 1.640&2.56&3.92&\textbf{0.53}& 0.82& 1.06& 1.29& 1.59& 2.11\\
                                                 &  sources  &\textbf{0.54}& 0.57& 1.530&2.87&4.26& 0.57& 0.85& 1.07& 1.41& 1.64& 2.26\\
                                                 &  proteins &\textbf{0.60}& 0.67& 4.190&5.46&3.24& 0.66& 0.96& 1.51& 1.79& 2.17& 2.75\\
\cmidrule(r){1-2}
\cmidrule(r){3-9}
\cmidrule(r){10-11}
\cmidrule(r){12-13}
\multirow{5}{*}{\rotatebox[origin=c]{90}{50~MB}} &  dna      &\textbf{2.02}& 2.360& 3.240&16.25&14.43& 2.06& 2.96& 3.88& 4.57& 5.94& 7.53\\
                                                 &  english  &\textbf{1.70}& 2.080& 65.85&15.41&12.76& 1.88& 2.65& 3.83& 4.56& 5.71& 7.21\\
                                                 &  dblp.xml & 1.41& 1.45& 4.370&9.490&9.370&\textbf{1.39}& 2.17& 2.93& 3.53& 4.32& 5.70\\
                                                 &  sources  &\textbf{1.45}& 1.49& 6.950&10.06&10.15& 1.51& 2.26& 2.87& 3.77& 4.38& 6.03\\
                                                 &  proteins &\textbf{1.77}& 2.01& 6.560&14.38&15.74& 1.87& 2.83& 4.55& 5.27& 6.42& 8.10\\
\cmidrule(r){1-2}
\cmidrule(r){3-9}
\cmidrule(r){10-11}
\cmidrule(r){12-13}
\multirow{5}{*}{\rotatebox[origin=c]{90}{100~MB}} &  dna      &\textbf{4.11}& 4.75& 6.590&26.03&26.62& 4.24& 5.95& 8.23& 9.44& 12.47& 15.39\\
                                                  &  english  &\textbf{3.56}& 4.28& 185.9&32.57&28.09& 4.02& 5.62& 7.96& 9.49& 11.98& 15.11\\
                                                  &  dblp.xml & 2.85& 2.89& 9.040&19.91&21.49&\textbf{2.82}& 4.41& 6.19& 7.22& 9.010& 11.63\\
                                                  &  sources  &\textbf{2.93}& 3.02& 39.85&24.92&24.46& 3.07& 4.62& 5.98& 7.72& 9.050& 12.34\\
                                                  &  proteins &\textbf{3.56}& 4.09& 16.99&30.89&28.12& 3.96& 5.86& 9.91& 10.96& 13.87& 16.82\\
\cmidrule(r){1-2}
\cmidrule(r){3-9}
\cmidrule(r){10-11}
\cmidrule(r){12-13}
\multirow{5}{*}{\rotatebox[origin=c]{90}{200~MB}} &  dna      &\textbf{8.25}& 10.0& 17.36&76.11&79.02& 8.64& 12.02& 17.41& 19.18& 26.05& 31.20\\
                                                  &  english  &\textbf{7.23}& 8.70& 1070&72.58&73.75& 8.25& 11.49& 16.80& 19.39& 25.05& 30.88\\
                                                  &  dblp.xml &\textbf{5.75}& 6.28& 18.23&49.97&52.91& 5.77& 9.120& 12.99& 14.72& 18.76& 23.84\\
                                                  &  sources  &\textbf{5.98}  & 6.23& 52.60&61.61&59.01& 6.37& 9.700& 12.63& 16.01& 19.00& 25.71\\
                                                  &  proteins &\textbf{6.86}& 7.94& 42.60&78.78&77.40& 8.33& 11.82& 19.73& 21.65& 28.06& 33.47\\
\bottomrule
\end{tabular}
\caption{The first seven columns contain the times solely for the computation of \LCP{}.
Since the inducing algorithms are interleaved with the computation of \SA{}, we subtracted the time to compute \SA{} with the corresponding inducing approach (``inducing [this paper]'' and ``inducing~\cite{Fischer2011}'').
\emph{GO} and \emph{GO2} require the \textsf{BWT} in addition to \SA{}; the time to compute \textsf{BWT} is also not included.
The last two columns show the time to compute \SA{} and \LCP{} using the inducing approach.
All times are in seconds, and are the average over 21 runs on the same input.}
\label{tab:time_evaluation}
\vspace{-2em}
\end{table}


\section{Conclusions} 
\label{sec:conclusions}
We presented a detailed description of \emph{DivSufSort} that has not been available albeit its wide use in different applications.
We linked interesting approaches, e.g., the repetition detection, to the corresponding lines in the source code and to the original literature.

Compared with SAIS, the other popular suffix array construction algorithm based on inducing, DivSufSort is faster.
We ascribe this to the two main differences between DivSufSort and SAIS:
First, the sorting of the initial suffixes in SAIS (the ones that cannot be induced) is done by recursively applying the algorithm (and renaming the initial suffixes), which is slower in practice than the string-sorting and prefix doubling-like approach used by DivSufSort (which also employs techniques like repetition detection to further decrease runtime).
Second, the classification of the initial suffixes differs:
while the suffixes that have to be sorted initially in SAIS can be displaced during the inducing of the \SA{}, they are not moved again in DivSufSort.
This also allows DivSufSort to skip parts (containing only \TypeA-suffixes) of the \SA{} during the first induction phase.

In addition, we showed that the LCP-array can be computed during the inducing of the suffix array in DivSufSort.
This approach is faster than the previous known inducing LCP-construction algorithm based on SAIS~\cite{Fischer2011}, and competitive with the $\Phi$-algorithm, i.e, the fastest pure LCP-construction algorithms.


\bibliographystyle{psc}
\bibliography{lit}

\begin{thebibliography}{10}

\bibitem{Bentley1997}
{\sc J.~L. Bentley and R.~Sedgewick}:
\newblock {\em Fast algorithms for sorting and searching strings}, in SODA,
  ACM/SIAM, 1997, pp.~360--369.

\bibitem{Bingmann2013a}
{\sc T.~Bingmann, J.~Fischer, and V.~Osipov}:
\newblock {\em Inducing suffix and lcp arrays in external memory.}, in ALENEX,
  SIAM, 2013, pp.~88--102.

\bibitem{dhaliwal12trends}
{\sc J.~Dhaliwal, S.~J. Puglisi, and A.~Turpin}:
\newblock {\em Trends in suffix sorting: {A} survey of low memory algorithms},
  in Proc.\ ACSC, Australian Computer Society, 2012, pp.~91--98.

\bibitem{Fischer2011}
{\sc J.~Fischer}:
\newblock {\em Inducing the {LCP}-array}, in Proc.\ WADS, vol.~6844 of LNCS,
  Springer, 2011, pp.~374--385.

\bibitem{Gog2014}
{\sc S.~Gog, T.~Beller, A.~Moffat, and M.~Petri}:
\newblock {\em From theory to practice: Plug and play with succinct data
  structures}, in Proc.\ SEA, vol.~8504 of LNCS, Springer, 2014, pp.~326--337.

\bibitem{Gog2011}
{\sc S.~Gog and E.~Ohlebusch}:
\newblock {\em Fast and lightweight {LCP}-array construction algorithms}, in
  Proc.\ ALENEX, SIAM, 2011, pp.~25--34.

\bibitem{Itoh1999}
{\sc H.~Itoh and H.~Tanaka}:
\newblock {\em An efficient method for in memory construction of suffix
  arrays}, in Proc.\ SPIRE/CRIWG, IEEE Press, 1999, pp.~81--88.

\bibitem{Karkkainen2009a}
{\sc J.~K\"arkk\"ainen, G.~Manzini, and S.~J. Puglisi}:
\newblock {\em Permuted longest-common-prefix array}, in Proc.\ CPM, vol.~5577
  of LNCS, Springer, 2009, pp.~181--192.

\bibitem{Kasai2001}
{\sc T.~Kasai, G.~Lee, H.~Arimura, S.~Arikawa, and K.~Park}:
\newblock {\em Linear-time longest-common-prefix computation in suffix arrays
  and its applications}, in Proc.\ CPM, vol.~2089 of LNCS, Springer, 2001,
  pp.~181--192.

\bibitem{Labeit2016parallelWaveletTree}
{\sc J.~Labeit, J.~Shun, and G.~E. Blelloch}:
\newblock {\em Parallel lightweight wavelet tree, suffix array and fm-index
  construction}, in Data Compression Conference (DCC), {IEEE}, 2016,
  pp.~33--42.

\bibitem{Larsson2007}
{\sc N.~J. Larsson and K.~Sadakane}:
\newblock {\em Faster suffix sorting}.
\newblock Theor. Comput. Sci., 387(3) 2007, pp.~258--272.

\bibitem{Manber1993}
{\sc U.~Manber and G.~Myers}:
\newblock {\em Suffix arrays: a new method for on-line string searches}.
\newblock siam Journal on Computing, 22(5) 1993, pp.~935--948.

\bibitem{Maniscalco2007}
{\sc M.~A. Maniscalco and S.~J. Puglisi}:
\newblock {\em An efficient, versatile approach to suffix sorting}.
\newblock ACM J.\ Experimental Algorithmics, 12 2008, p.~Article no.\ 1.2.

\bibitem{Manzini2004}
{\sc G.~Manzini and P.~Ferragina}:
\newblock {\em Engineering a lightweight suffix array construction algorithm}.
\newblock Algorithmica, 40(1) 2004, pp.~33--50.

\bibitem{Mehlhorn1984}
{\sc K.~Mehlhorn}:
\newblock {\em Data Structures and Algorithms 1: Sorting and Searching}, vol.~1
  of {EATCS} Monographs on Theoretical Computer Science, Springer, 1984.

\bibitem{Musser1997}
{\sc D.~R. Musser}:
\newblock {\em Introspective sorting and selection algorithms}.
\newblock Softw., Pract. Exper., 27(8) 1997, pp.~983--993.

\bibitem{Nong2009}
{\sc G.~Nong, S.~Zhang, and W.~H. Chan}:
\newblock {\em Linear suffix array construction by almost pure
  induced-sorting}, in Proc.\ DCC, IEEE Press, 2009, pp.~193--202.

\bibitem{Puglisi2007}
{\sc S.~J. Puglisi, W.~F. Smyth, and A.~H. Turpin}:
\newblock {\em A taxonomy of suffix array construction algorithms}.
\newblock ACM Comput. Surv., 39(2) 2007.

\end{thebibliography}
\end{document}